\documentclass[10pt]{article}
\usepackage[utf8]{inputenc}
\usepackage{style}
\usepackage{framed}
\usepackage[parfill]{parskip}
\usepackage{setspace}
\usepackage{hyperref}
\allowdisplaybreaks

\title{Inapproximability of Maximum Biclique Problems, \\ Minimum $k$-Cut and Densest At-Least-$k$-Subgraph from \\ the Small Set Expansion Hypothesis\footnote{An extended abstract of this work will appear at ICALP 2017 under a different title~\cite{Man17-ICALP}.}}

\author{Pasin Manurangsi\thanks{Email: \texttt{pasin@berkeley.edu}. This material is based upon work supported by the National Science Foundation under Grants No. CCF 1540685 and CCF 1655215} \vspace{-0.5em}\\
UC Berkeley}

\begin{document}

\maketitle
\thispagestyle{empty}

\begin{abstract}
The \emph{Small Set Expansion Hypothesis} is a conjecture which roughly states that it is NP-hard to distinguish between a graph with a small subset of vertices whose (edge) expansion is almost zero and one in which all small subsets of vertices have expansion almost one. In this work, we prove conditional inapproximability results for the following graph problems based on this hypothesis:
\begin{itemize}
\item \emph{Maximum Edge Biclique}: given a bipartite graph $G$, find a complete bipartite subgraph of $G$ that contains maximum number of edges. We show that, assuming the Small Set Expansion Hypothesis and that NP $\nsubseteq$ BPP, no polynomial time algorithm approximates Maximum Edge Biclique to within a factor of $n^{1 - \varepsilon}$ of the optimum for every constant $\varepsilon > 0$.
\item \emph{Maximum Balanced Biclique}: given a bipartite graph $G$, find a balanced complete bipartite subgraph of $G$ that contains maximum number of vertices. Similar to Maximum Edge Biclique, we prove that this problem is inapproximable in polynomial time to within a factor of $n^{1 - \varepsilon}$ of the optimum for every constant $\varepsilon > 0$, assuming the Small Set Expansion Hypothesis and that NP $\nsubseteq$ BPP.
\item \emph{Minimum $k$-Cut}: given a weighted graph $G$, find a set of edges with minimum total weight whose removal partitions the graph into (at least) $k$ connected components. For this problem, we prove that it is NP-hard to approximate it to within $(2 - \varepsilon)$ factor of the optimum for every constant $\varepsilon > 0$, assuming the Small Set Expansion Hypothesis.
\item \emph{Densest At-Least-$k$-Subgraph}: given a weighted graph $G$, find a set $S$ of at least $k$ vertices such that the induced subgraph on $S$ has maximum \emph{density}, which is defined as the ratio between the total weight of the edges and the number of the vertices. We show that this problem is NP-hard to approximate to within $(2 - \varepsilon)$ factor of the optimum for every constant $\varepsilon > 0$, assuming the Small Set Expansion Hypothesis.
\end{itemize}
The ratios in our inapproximability results are essentially tight since trivial algorithms give $n$-approximation to both Maximum Edge Biclique and Maximum Balanced Biclique and polynomial time $2$-approximation algorithms are known for Minimum $k$-Cut~\cite{SV95} and Densest At-Least-$k$-Subgraph~\cite{And07,KS09}.

Our first two results are proved by combining a technique developed by Raghavendra, Steurer and Tulsiani~\cite{RST12} to avoid locality of gadget reductions with a generalization of Bansal and Khot's long code test~\cite{BK09} whereas our last two results are shown via elementary reductions.
\end{abstract}

\newpage
\setcounter{page}{1}

\section{Introduction}

Since the PCP theorem was proved two decades ago~\cite{ALMSS98,AS98}, our understanding of approximability of combinatorial optimization problems has grown enormously; tight inapproximability results have been obtained for fundamental problems such as Max-3SAT~\cite{Has01}, Max Clique~\cite{Has96} and Set Cover~\cite{Mos15,DS14}. Yet, for other problems, including Vertex Cover and Max Cut, known NP-hardness of approximation results come short of matching best known algorithms.

Khot's introduction of the Unique Games Conjecture (UGC)~\cite{Kho02} propelled another wave of development in hardness of approximation that saw many of these open problems resolved (see e.g.~\cite{KR08,KKMO07}). Alas, some problems continue to elude even attempts at proving UGC-hardness of approximation. For a class of such problems, the failure stems from the fact that typical reductions are local in nature; many reductions from unique games to graph problems could produce disconnected graphs. If we try to use such reductions for problems that involve some forms of expansion of graphs (e.g. Sparsest Cut), we are out of luck.

One approach to overcome the aforementioned issue is through the \emph{Small Set Expansion Hypothesis} (SSEH) of Raghavendra and Steurer~\cite{RS10}. To describe the hypothesis, let us introduce some notations. Throughout the paper, we represent an undirected edge-weighted graph $G = (V, E, w)$ by a vertex set $V$, an edge set $E$ and a weight function $w: E \to \R_{\geqs 0}$. We call $G$ \emph{$d$-regular} if $\sum_{v: (u, v) \in E} w(u, v) = d$ for every $u \in V$. For a $d$-regular weighted graph $G$, the edge expansion $\Phi(S)$ of $S \subseteq V$ is defined as $$\Phi(S) = \frac{E(S, V \setminus S)}{d\min\{|S|, |V \setminus S|\}}$$ where $E(S, V \setminus S)$ is the total weight of edges across the cut $(S, V \setminus S)$. The small set expansion problem SSE$(\delta, \eta)$, where $\eta, \delta$ are two parameters that lie in (0, 1), can be defined as follows.

\begin{definition}[SSE($\delta, \eta$)]
Given a regular edge-weighted graph $G = (V, E, w)$, distinguish between:
\begin{itemize}
\item \emph{(Completeness)} There exists $S \subseteq V$ of size $\delta |V|$ such that $\Phi(S) \leqs \eta$.
\item \emph{(Soundness)} For every $S \subseteq V$ of size $\delta |V|$, $\Phi(S) \geqs 1 - \eta$.
\end{itemize}
\end{definition}

Roughly speaking, SSEH asserts that it is NP-hard to distinguish between a graph that has a small non-expanding subset of vertices and one in which all small subsets of vertices have almost perfect edge expansion. More formally, the hypothesis can be stated as follows.

\begin{conjecture}[SSEH~\cite{RS10}] \label{conj:sse}
For every $\eta > 0$, there is $\delta = \delta(\eta) > 0$ such that SSE($\delta, \eta$) is NP-hard.
\end{conjecture}

Interestingly, SSEH not only implies UGC~\cite{RS10}, but it is also equivalent to a strengthened version of the latter, in which the graph is required to have almost perfect small set expansion~\cite{RST12}.

Since its proposal, SSEH has been used as a starting point for inapproximability of many problems whose hardnesses are not known otherwise. Most relevant to us is the work of Raghavendra, Steurer and Tulsiani (henceforth RST)~\cite{RST12} who devised a technique that exploited structures of SSE instances to avoid locality in reductions. In doing so, they obtained inapproximability of Min Bisection, Balanced Separator, and Minimum Linear Arrangement, which are not known to be hard to approximate under UGC.

\subsection{Maximum Edge Biclique and Maximum Balanced Biclique}

Our first result is adapting RST technique to prove inapproximability of Maximum Edge Biclique (MEB) and Maximum Balanced Biclique (MBB). For both problems, the input is a bipartite graph. The goal for the former is to find a complete bipartite subgraph that contains as many edges as possible whereas, for the latter, the goal is to find a balanced complete bipartite subgraph that contains as many vertices as possible.

Both problems are NP-hard. MBB was stated (without proof) to be NP-hard in~\cite[page 196]{GJ79}; several proofs of this exist such as one provided in~\cite{Joh87}. For MEB, it was proved to be NP-hard more recently by Peeters~\cite{Pee03}. Unfortunately, much less is known when it comes to approximability of both problems. Similar to Maximum Clique, folklore algorithms give $O(n/\polylog n)$ approximation ratio for both MBB and MEB, and no better algorithm is known. However, not even NP-hardness of approximation of some constant ratio is known for the problems. This is in stark contrast to Maximum Clique for which strong inapproximability results are known~\cite{Has96,Kho01,KP06,Zuc07}. Fortunately, the situation is not completely hopeless as the problems are known to be hard to approximate under stronger complexity assumptions.

Feige~\cite{Fei02} showed that, assuming that random 3SAT formulae cannot be refuted in polynomial time, both problems\footnote{While Feige only stated this for MBB, the reduction clearly works for MEB too.} cannot be approximated to within $n^\varepsilon$ of the optimum in polynomial time for some $\varepsilon > 0$. Later, Feige and Kogan~\cite{FK04} proved $2^{(\log n)^\varepsilon}$ ratio inapproximability for both problems for some $\varepsilon > 0$, assuming that 3SAT $\notin$ DTIME($2^{n^{3/4 + \delta}}$) for some $\delta > 0$. Moreover, Khot~\cite{Kho06} showed, assuming 3SAT $\notin$ BPTIME($2^{n^\delta}$) for some $\delta > 0$, that no polynomial time algorithm achieves $n^\varepsilon$-approximation for MBB for some $\varepsilon > 0$. Amb\"{u}hl \etal~\cite{AMS11} subsequently built on Khot's result and showed a similar hardness for MEB. Recently, Bhangale \etal~\cite{BGHKK16} proved that both problems are hard to approximate to within\footnote{In~\cite{BGHKK16}, the inapproximability ratio is only claimed to be $n^{\varepsilon}$ for \emph{some} $\varepsilon > 0$. However, it is not hard to see that their result in fact implies $n^{1 - \varepsilon}$ factor hardness of approximation as well.} $n^{1 - \varepsilon}$ factor for every $\varepsilon > 0$, assuming a certain strengthened version of UGC and NP $\ne$ BPP. In addition, while not stated explicitly, the author's recent reduction for Densest $k$-Subgraph~\cite{Man17} yields $n^{1/\polyloglog n}$ ratio inapproximability for both problems under the Exponential Time Hypothesis~\cite{IPZ01} (3SAT $\notin$ DTIME($2^{o(n)}$)) and this ratio can be improved to $n^{f(n)}$ for any $f \in o(1)$ under the stronger Gap Exponential Time Hypothesis~\cite{Din16,MR16} (no $2^{o(n)}$ time algorithm can distinguish a fully satisfiable 3SAT formula from one which is only $(1 - \varepsilon)$-satisfiable for some $\varepsilon > 0$); these ratios are better than those in~\cite{FK04} but worse than those in~\cite{Kho06,AMS11,BGHKK16}.

Finally, it is worth noting that, assuming the Planted Clique Hypothesis~\cite{Jer92,Kuc95} (no polynomial time algorithm can distinguish between a random graph $\cG(n, 1/2)$ and one with a planted clique of size $\Omega(\sqrt{n})$), it follows (by partitioning the vertex set into two equal sets and delete all the edges within each partition) that Maximum Balanced Biclique cannot be approximated to within $\tO(\sqrt{n})$ ratio in polynomial time. Interestingly, this does not give any hardness for Maximum Edge Biclique, since the planted clique has only $O(n)$ edges, which less than that in a trivial biclique consisting of any vertex and all of its neighbors.

In this work, we prove strong inapproximability results for both problems, assuming SSEH:

\begin{theorem} \label{thm:meb-mbb}
Assuming SSEH, there is no polynomial time algorithm that approximates MEB or MBB to within $n^{1 - \varepsilon}$ factor of the optimum for every $\varepsilon > 0$, unless NP $\subseteq$ BPP.
\end{theorem}

We note that the only part of the reduction that is randomized is the gap amplification via randomized graph product~\cite{BS92,Blu91}. If one is willing to assume only that NP $\ne$ P (and SSEH), our reduction still implies that both are hard to approximate to within any constant factor.

Only Bhangale \etal's result~\cite{BGHKK16} and our result achieve the inapproximability ratio of $n^{1 - \varepsilon}$ for every $\varepsilon > 0$; all other results achieve at most $n^\varepsilon$ ratio for some $\varepsilon > 0$. Moreover, only Bhangale \etal's reduction and ours are candidate NP-hardness reductions, whereas each of the other reductions either uses superpolynomial time~\cite{FK04,Kho06,AMS11,Man17} or relies on an average-case assumption~\cite{Fei02}. It is also worth noting here that, while both Bhangale \etal's result and our result are based on assumptions which can be viewed as stronger variants of UGC, the two assumptions are incomparable and, to the best of our knowledge, Bhangale \etal's technique does not apply to SSEH. A discussion on the similarities and differences between the two assumptions can be found in Appendix~\ref{app:ugc}.

Along the way, we prove inapproximability of the following hypergraph bisection problem, which may be of independent interest: given a hypergraph $H = (V_H, E_H)$ find a bisection\footnote{$(T_0, T_1)$ is a bisection of $V_H$ if $|T_0| = |T_1| = |V_H|/2, T_0 \cap T_1 = \emptyset$ and $V_H = T_0 \cup T_1$.} $(T_0, T_1)$ of $V_H$ such that the number of uncut hyperedges is maximized. We refer to this problem as \emph{Max UnCut Hypergraph Bisection} (MUCHB). Roughly speaking, we show that, assuming SSEH, it is hard to distinguish a hypergraph whose optimal bisection cuts only $\varepsilon$ fraction of hyperedges from one in which every bisection cuts all but $\varepsilon$ fraction of hyperedges:

\begin{lemma} \label{lem:hypergraph-bisec}
Assuming SSEH, for every $\varepsilon > 0$, it is NP-hard to, given a hypergraph $H = (V_H, E_H)$, distinguish between the following two cases:
\begin{itemize}
\item \emph{(Completeness)} There is a bisection $(T_0, T_1)$ of $V_H$ s.t. $|E_H(T_0)|, |E_H(T_1)| \geqs (1/2 - \varepsilon)|E_H|$.
\item \emph{(Soundness)} For every set $T \subseteq V_H$ of size at most $|V_H|/2$, $|E_H(T)| \leqs \varepsilon|E_H|$.
\end{itemize}
Here $E_H(T) \triangleq \{e \in E_H \mid e \subseteq T\}$ denotes the set of hyperedges that lie completely inside of the set $T \subseteq V_H$.
\end{lemma}

Our result above is similar to Khot's quasi-random PCP~\cite{Kho06}. Specifically, Khot's quasi-random PCP can be viewed as a hardness for MUCHB in the setting where the hypergraph is $d$-uniform; 
roughly speaking, Khot's result states that it is hard (if 3SAT $\notin$ $\bigcap_{\delta > 0}$ BPTIME($2^{n^\delta}$)) to distinguish between a $d$-uniform hypergraph where $1/2^{d-2}$ fraction of hyperedges are uncut in the optimal bisection from one where roughly $1/2^{d-1}$ fraction of hyperedges are uncut in any bisection. Note that the latter is the fraction of uncut hyperedges in random hypergraphs and hence the name ``quasi-random''. 
In this sense, Khot's result provides better soundness at the expense of worse completeness compared to Theorem~\ref{lem:hypergraph-bisec}.

\subsection{Minimum $k$-Cut}

In addition to the above biclique problems, we prove an inapproximability result for the Minimum $k$-Cut problem, in which a weighted graph is given and the goal is to find a set of edges with minimum total weight whose removal paritions the graph into (at least) $k$ connected components. The Minimum $k$-Cut problem has long been studied. 
When $k = 2$, the problem can be solved in polynomial time simply by solving Minimum $s-t$ cut for every possible pairs of $s$ and $t$. In fact, for any fixed $k$, the problem was proved to be in P by Goldschmidt and Hochbaum~\cite{GH94}, who also showed that, when $k$ is part of the input, the problem is NP-hard. To circumvent this, Saran and Vazirani~\cite{SV95} devised two simple polynomial time $(2 - 2/k)$-approximation algorithms for the problem. In the ensuing years, different approximation algorithms~\cite{NR01,ZNI01,RS02,XCY11} have been proposed for the problem, none of which are able achieve an approximation ratio of $(2 - \varepsilon)$ for some $\varepsilon > 0$. In fact, Saran and Vazirani themselves conjectured that $(2 - \varepsilon)$-approximation is intractible for the problem~\cite{SV95}. In this work, we show that their conjecture is indeed true, if the SSEH holds:

\begin{theorem} \label{thm:k-cut}
Assuming SSEH, it is NP-hard to approximate Minimum $k$-Cut to within $(2 - \varepsilon)$ factor of the optimum for every constant $\varepsilon > 0$.
\end{theorem}

Note that the problem was claimed to be APX-hard in~\cite{SV95}. However, to the best of our knowledge, the proof has never been published and no other inapproximability is known.

\subsection{Densest At-Least-$k$-Subgraph}

Our final result is a hardness of approximating the Densest At-Least-$k$-Subgraph (DAL$k$S) problem, which can be stated as follows. Given a weighted graph, find a subset $S$ of at least $k$ vertices such that the induced subgraph on $S$ has maximum \emph{density}, which is defined as the ratio between the total weight of edges and the number of vertices. The problem was first introduced by Andersen and Chellapilla~\cite{AC09} who also gave a 3-approximation algorithm for the problem. Shortly after, 2-approximation algorithms for the problem were discovered by Andersen~\cite{And07} and independently by Khuller and Saha~\cite{KS09}. We show that, assuming SSEH, this approximation guarantee is essentially the best we can hope for:

\begin{theorem} \label{thm:dalks}
Assuming the Small Set Expansion Hypothesis, it is NP-hard to approximate Densest At-Least-$k$-Subgraph to within $(2 - \varepsilon)$ factor of the optimum for every constant $\varepsilon > 0$.
\end{theorem}

To the best of our knowledge, no hardness of approximation for DAL$k$S was explicitly proved before. We remark that DAL$k$S is a variant of the Densest $k$-Subgraph (D$k$S) problem, which is the same as DAL$k$S except that the desired set $S$ must have size exactly $k$. D$k$S has been extensively studied dating back to the early 90s~\cite{KP93,FS97,SW98,FL01,FPK01,AHI02,Fei02,Kho06,GL09,RS10,BCCFV10,AAMMW11,BCVGZ12,BKRW17,Man17}. Despite these considerable efforts, its approximability is still wide open. In particular, even though lower bounds have been shown under stronger complexity assumptions~\cite{Fei02,Kho06,RS10,AAMMW11,BKRW17,Man17} and for LP/SDP hierarchies~\cite{BCCFV10,M-thesis,CMMV17}, not even constant factor NP-hardness of approximation for D$k$S is known. On the other hand, the best polynomial time algorithm for D$k$S achieves only $O(n^{1/4 + \varepsilon})$-approximation~\cite{BCCFV10}. Since any inapproximability result for DAL$k$S translates directly to D$k$S, even proving some constant factor NP-hardness of approximating DAL$k$S would advance our understanding of approximability of D$k$S.

\section{Inapproximability of Minimum $k$-Cut}

We now proceed to prove our main results. Let us start with the simplest: Minimum $k$-Cut.

\begin{proof}[Proof of Theorem~\ref{thm:k-cut}]
The reduction from SSE($\delta, \eta$) to Minimum $k$-Cut is simple; the graph $G$ remains the input graph for Minimum $k$-Cut and we let $k = \delta n + 1$ where $n = |V|$.

{\bf Completeness.} If there is $S \subseteq V$ of size $\delta n$ such that $\Phi(S) \leqs \eta$, then we partition the graph into $k$ groups where the first group is $V \setminus S$ and each of the other groups contains one vertex from $S$. The edges cut are the edges across the cut $(S, V \setminus S)$ and the edges within the set $S$ itself. The total weight of edges of the former type is $d|S|\Phi(S) \leqs \eta d |S|$ and that of the latter type is at most $d|S|/2$. Hence, the total weight of edges cut in this partition is at most $(1/2 + \eta) d |S| = (1/2 + \eta) \delta d n$.

{\bf Soundness.} Suppose that, for every $S \subseteq V$ of size $\delta n$, $\Phi(S) \geqs 1 - \eta$. Let $T_1, \dots, T_k \subseteq V$ be any $k$-partition of the graph. Assume without loss of generality that $|T_1| \leqs \cdots \leqs |T_k|$. Let $A = T_1 \cup \cdots \cup T_i$ where $i$ is the maximum index such that $|T_1 \cup \cdots \cup T_i| \leqs \delta n$.

We claim that $|A| \geqs \delta n - \sqrt{n}$. To see that this is the case, suppose for the sake of contradiction that $|A| < \delta n - \sqrt{n}$. Since $|A \cup T_{i + 1}| > \delta n$, we have $T_{i + 1} > \sqrt{n}$. Moreover, since $A = T_1 \cup \cdots T_i$, we have $i \leqs |A| < \delta n - \sqrt{n}$. As a result, we have $n = |T_1 \cup \cdots \cup T_{k}| \geqs |T_{i + 1} \cup \cdots \cup T_{k}| \geqs (k - i)|T_i| > \sqrt{n} \cdot \sqrt{n} = n$, which is a contradiction. Hence, $|A| \geqs \delta n - \sqrt{n}$.

Now, note that, for every  $S \subseteq V$ of size $\delta n$, $\Phi(S) \geqs 1 - \eta$ implies that $E(S) \leqs \eta d \delta n / 2$ where $E(S)$ denote the total weight of all edges within $S$. Since $|A| \leqs \delta n$, we also have $E(A) \leqs \eta d \delta n / 2$. As a result, the total weight of edges across the cut $(A, V \setminus A)$, all of which are cut by the partition, is at least $$d|A| - \eta d \delta n \geqs (1 - \eta) d \delta n - d \sqrt{n} = \left(1 - \eta - \frac{1}{\delta \sqrt{n}}\right)\delta dn.$$

For every sufficiently small constant $\varepsilon > 0$, by setting $\eta = \varepsilon/20$ and $n \geqs 100/(\varepsilon^2 \delta^2)$, the ratio between the two cases is at least $(2 - \varepsilon)$, which concludes the proof of Theorem~\ref{thm:k-cut}.
\end{proof}

\section{Inapproximability of Densest At-Least-$k$-Subgraph}

We next prove our inapproximability result for Densest At-Least-$k$-Subgraph, which is also very simple. For this reduction and the subsequent reductions, it will be more convenient for us to use a different (but equivalent) formulation of SSEH. To state it, we first define a variant of SSE($\delta, \eta$) called SSE($\delta, \eta, M$); the completeness remains the same whereas the soundness is strengthened to include all $S$ of size in $\left[\frac{\delta |V|}{M}, \delta |V| M\right]$.

\begin{definition}[SSE($\delta, \eta, M$)]
Given a regular edge-weighted graph $G = (V, E, w)$, distinguish between:
\begin{itemize}
\item \emph{(Completeness)} There exists $S \subseteq V$ of size $\delta |V|$ such that $\Phi(S) \leqs \eta$.
\item \emph{(Soundness)} For every $S \subseteq V$ with $|S| \in \left[\frac{\delta |V|}{M}, \delta |V| M\right]$, $\Phi(S) \geqs 1 - \eta$.
\end{itemize}
\end{definition}

The new formulation of the hypothesis can now be stated as follows.

\begin{conjecture} \label{conj:sse-strong}
For every $\eta, M > 0$, there is $\delta = \delta(\eta, M) > 0$ such that SSE($\delta, \eta, M$) is NP-hard.
\end{conjecture}

Raghavendra \etal~\cite{RST12} showed that this formulation is equivalent to the earlier formulation (Conjecture~\ref{conj:sse}); please refer to Appendix A.2 of~\cite{RST12} for a simple proof of this.

\begin{proof}[Proof of Theorem~\ref{thm:dalks}]
Given an instance $G = (V, E, w)$ of SSE($\delta, \eta, M$), we create an input graph $G' = (V', w')$ for Densest At-Least-$k$-Subgraph as follows. $V'$ consists of all the vertices in $V$ and an additional vertex $v^*$. The weight function $w'$ remains the same as $w$ for all edges in $V$ whereas $v^*$ has only a self-loop\footnote{If we would like to avoid self-loops, we can replace $v^*$ with two vertices $v_1^*, v_2^*$ with an edge of weight $d \delta n / 2$ between them.} with weight $d \delta n / 2$. In other words, $E' = E \cup \{v^*\}$ and $w'(v^*) = d \delta n / 2$. Finally, let $k = 1 + \delta n$ where $n = |V|$.

{\bf Completeness.} If there is $S \subseteq V$ of size $\delta n$ such that $\Phi(S) \leqs \eta$, consider the set $S' = S \cup \{v^*\}$. We have $|S'| = k$ and the density of $S'$ is $\left(d \delta n / 2 + E(S)\right)/k$ where $E(S)$ denote the total weight of edges within $S$. This can be written as $$\left(\frac{\delta n}{k}\right)\left(\frac{d \delta n / 2 + E(S)}{\delta n}\right) = \left(\frac{\delta n}{k}\right)\left(\frac{d}{2} + \left(\frac{1 - \Phi(S)}{2}\right)d\right) \geqs \frac{d\delta n(1 - \eta/2)}{k}.$$

{\bf Soundness.} Suppose that $\Phi(S) \geqs 1 - \eta$ for every $S \subseteq V$ of size $|S| \in [\delta n / M, \delta n M]$. Consider any set $T' \subseteq V$ of size at least $k$. Let $T = T' \setminus \{v^*\}$ and let $E(T)$ denote the total weight of edges within $T$. Observe that the density of $S$ is at most $(d\delta n/2 + E(T))/|T'|$. Let us consider the following two cases.
\begin{enumerate}
\item $|T| \leqs \delta n M$. In this case, $\Phi(T) \geqs 1 - \eta$ and we have $$\frac{d\delta n/2 + E(T)}{|T'|} \leqs \frac{d\delta n}{2k} + \left(\frac{1 - \Phi(T)}{2}\right) \leqs \frac{d\delta n}{2k} + d\eta/2 = \frac{d\delta n\left(\frac{1}{2} + \frac{\eta}{2}\left(1 + \frac{1}{\delta n}\right)\right)}{k} \leqs \frac{d\delta n(1/2 + \eta)}{k}.$$
\item $|T| > \delta n M$. In this case, we have
$$\frac{d\delta n/2 + E(T)}{|T'|} < \frac{d}{2M} + \frac{d}{2} = \frac{d\delta n}{k}\left(\frac{1}{2} + \frac{1}{2M}\right)\left(1 + \frac{1}{\delta n}\right) \leqs \frac{d\delta n\left(\frac{1}{2} + \frac{1}{\delta n} + \frac{1}{M}\right)}{k}.$$
\end{enumerate}
Hence, in both cases, the density of $T'$ is at most $d\delta n\left(1/2 + \max\{\eta, \frac{1}{\delta n} + \frac{1}{M}\}\right)/k$.

For every sufficiently small constant $\varepsilon > 0$, by picking $\eta = \varepsilon/20, M = 40/\varepsilon$ and $n \geqs 800/(\varepsilon \delta)$, the ratio between the two cases is at least $(2 - \varepsilon)$, concluding the proof of Theorem~\ref{thm:dalks}.
\end{proof}

\section{Inapproximability of MEB and MBB} \label{sec:meb-mbb}

Let us now turn our attention to MEB and MBB. First, note that we can reduce MUCHB to MEB/MBB by just letting the two sides of the bipartite graph be $E_H$ and creating an edge $(e_1, e_2)$ iff $e_1 \cap e_2 = \emptyset$. This immediately shows that Lemma~\ref{lem:hypergraph-bisec} implies the following:

\begin{lemma} \label{lem:meb-mbb-withoutamp}
Assuming SSEH, for every $\delta > 0$, it is NP-hard to, given a bipartite graph $G = (L, R, E)$ with $|L| = |R| = n$, distinguish between the following two cases:
\begin{itemize}
\item \emph{(Completeness)} $G$ contains $K_{(1/2 - \delta)n, (1/2 - \delta)n}$ as a subgraph.
\item \emph{(Soundness)} $G$ does not contain $K_{\delta n, \delta n}$ as a subgraph.
\end{itemize}
Here $K_{t, t}$ denotes the complete bipartite graph in which each side contains $t$ vertices.
\end{lemma}

We provide the full proof of Lemma~\ref{lem:meb-mbb-withoutamp} in Appendix~\ref{app:red-mchb}. We also note that Theorem~\ref{thm:meb-mbb} follows from Lemma~\ref{lem:meb-mbb-withoutamp} by gap amplification via randomized graph product~\cite{BS92,Blu91}. Since this has been analyzed before even for biclique~\cite[Appendix D]{Kho06}, we defer the full proof to Appendix~\ref{app:amp}.

We are now only left to prove Lemma~\ref{lem:hypergraph-bisec}; we devote the rest of this section to this task.

\subsection{Preliminaries}

Before we continue, we need additional notations and preliminaries. For every graph $G = (V, E, w)$ and every vertex $v$, we write $G(v)$ to denote the distribution on its neighbors weighted according to $w$. Moreover, we sometimes abuse the notation and write $e \sim G$ to denote a random edge of $G$ weighted according to $w$.

While our reduction can be understood without notation of unique games, it is best described in a context of unique games reductions. We provide a definition of unique games below.

\begin{definition}[Unique Game (UG)]
A unique game instance $\cU = (\cG = (\cV, \cE, \cW), [R], \{\pi_e\}_{e \in \cE})$ consists of an edge-weighted graph $\cG = (\cV, \cE, \cW)$, a label set $[R] = \{1, \dots, R\}$, and, for each $e \in \cE$, a permutation $\pi_e: [R] \rightarrow [R]$. The goal is to find an assignment $F: V \rightarrow [R]$ such that $\val_{\cU}(F) \triangleq \Pr_{(u, v) \sim G}[\pi_{(u, v)}(F(u)) = F(v)]$ is maximized; we call an edge $(u, v)$ such that $\pi_{(u, v)}(F(u)) = F(v)$ \emph{satisfied}.
\end{definition}

Khot's UGC~\cite{Kho02} states that, for every $\varepsilon > 0$, it is NP-hard to distinguish between a unique game in which there exists an assignment satisfying at least $(1 - \varepsilon)$ fraction of edges from one in which every assignment satisfies at most $\varepsilon$ fraction of edges.

Finally, we need some preliminaries in discrete Fourier analysis. We state here only few facts that we need. We refer interested readers to~\cite{OD14} for more details about the topic.

For any discrete probability space $\Omega$, $f: \Omega^R \rightarrow [0, 1]$ can be written as $\sum_{\sigma \in [|\Omega|]^R} \hf(\sigma) \phi_\sigma$ where $\{\phi_\sigma\}_{\sigma \in [|\Omega|]^R}$ is the product Fourier basis of $L^2(\Omega^R)$ (see~\cite[Chapter 8.1]{OD14}). The degree-$d$ influence on the $j$-th coordinate of $f$ is $\infl_j^d(f) \triangleq \sum_{\sigma \in [|\Omega|]^R, \sigma_j \ne 1, \#\sigma \leqs d} \hf^2(\sigma)$ where $\#\sigma \triangleq |\{i \in [R] \mid \sigma_i \ne 1\}|$. It is well known that $\sum_{j=1}^{R} \infl_j^d(f) \leqs d$ (see~\cite[Proposition 3.8]{MOO10}). 

We also need the following theorem. It follows easily\footnote{For more details on how this version follows from there, please refer to~\cite[page 769]{Sve13}.} from the so-called ``It Ain't Over Till It's Over'' conjecture, which is by now a theorem~\cite[Theorem 4.9]{MOO10}. 
\begin{theorem}[\cite{MOO10}] \label{thm:aint-over}
For any $\beta, \varepsilon_T, \gamma > 0$, there exists $\kappa > 0$ and $t , d \in \mathbb{N}$ such that, if any functions $f_1, \dots, f_t: \Omega^R \rightarrow \{0, 1\}$ where $\Omega$ is a probability space whose probability of each atom is at least $\beta$ satisfy\footnote{0.99 could be replaced by any constant less than one; we use it to avoid introducing more parameters.}
\begin{align*}
\forall i \in [t], \E_{x \in \Omega^R}[f_i(x)] \leqs 0.99 & &\text{and} & &\forall j \in [R], \forall 1 \leqs i_1 \ne i_2 \leqs t, \min\{\infl_j^d(f_{i_1}), \infl_j^d(f_{i_2})\} \leqs \kappa,
\end{align*}
then
\begin{align*}
\Pr_{x \in \Omega^R, D \sim S_{\varepsilon_T}(R)}\left[\bigwedge_{i=1}^t f_i(C_D(x)) \equiv 1\right] < \gamma
\end{align*}
where $D \sim S_{\varepsilon_T}(R)$ is a random subset of $[R]$ where each $i \in [R]$ is included independently w.p. $\varepsilon_T$, $C_D(x) \triangleq \{x' \mid x'_{[R] \setminus D} = x_{[R] \setminus D}\}$ and $f_i(C_D(x)) \equiv 1$ is short for $\forall x'\in C_D(x), f_i(x') = 1$.
\end{theorem}

\subsection{Bansal-Khot Long Code Test and A Candidate Reduction} \label{subsec:bk}

Theorem~\ref{thm:aint-over} leads us nicely to the Bansal-Khot long code test~\cite{BK09}. For UGC hardness reductions, one typically needs a \emph{long code test} (aka dictatorship gadget) which, on input $f_1, \dots, f_t: \{0, 1\}^R \rightarrow \{0, 1\}$, has the following properties:
\begin{itemize}
\item \emph{(Completeness)} If $f_1 = \cdots = f_t$ is a long code\footnote{A long code is simply $j$-junta (i.e. a function that depends only on the $x_j$) for some $j \in [R]$.}, the test accepts with large probability.
\item \emph{(Soundness)} If $f_1, \dots, f_t$ are balanced (i.e. $\E f_1 = \cdots = \E f_t = 1/2$) and are ``far from being a long code'', then the test accepts with low probability. 

A widely-used notion of ``far from being a long code'', and one we will use here, is that the functions do not share a coordinate with large low degree influence (i.e. for every $j \in [R]$ and every $i_1 \ne i_2 \in [t]$, at least one of $\infl_j^d(f_{i_1})$, $\infl_j^d(f_{i_2})$ is small), 
\end{itemize}

Bansal-Khot long code test works by first picking $x \sim \{0, 1\}^R$ and $D \sim S_{\varepsilon_T}(R)$. Then, test whether $f_i$ evaluates to 1 on the whole $C_D(x)$. This can be viewed as an ``algorithmic'' version of Theorem~\ref{thm:aint-over}; specifically, the theorem (with $\Omega = \{0, 1\}$) immediately implies the soundness property of this test. On the other hand, it is obvious that, if $f_1 = \cdots = f_t$ is a long code, then the test accepts with probability $1/2 - \varepsilon_T$.

Bansal and Khot used this test to prove tight hardness of approximation of Vertex Cover. The reduction is via a natural composition of the test with unique games. Their reduction also gives a cadidate reduction from UG to MUCHB, which is stated below in Figure~\ref{fig:red1}.

\begin{figure}[h!]
\begin{framed}
Input: A unique game $(\cG = (\cV, \cE, \cW), [R], \{\pi_e\}_{e \in E})$ and parameters $\ell \in \mathbb{N}$ and $\varepsilon_T \in (0, 1)$. \\
Output: A hypergraph $H = (V_H, E_H)$. \\
The vertex set $V_H$ is $\cV \times \{0, 1\}^R$ and the hyperedges are distributed as follows:
\begin{itemize}
\item Sample $u \sim \cV$ and sample $v_1 \sim \cG(u), \dots, v_\ell \sim \cG(u)$.
\item Sample $x \sim \{0, 1\}^R$ and a subset $D \sim S_{\varepsilon_T}(R)$.
\item Output a hyperedge $e = \{(v_p, x') \mid p \in [\ell], x' \in C_D(x)\}$.
\end{itemize}
\end{framed}
\caption{A Candidate Reduction from UG to MUCHB}
\label{fig:red1}
\end{figure}

As is typical for gadget reductions, for $T \subseteq V_H$, we view the indicator function $f_u(x) \triangleq \mathds{1}[(u, x) \in T]$ for each $u \in \cV$ as the intended long code. If there exists an assignment $\phi$ to the unique game instance that satisfies nearly all the constraints, then the bisection corresponding to $f_u(x) = x_{\phi(u)}$ cuts only small fraction of edges, which yields the completeness of MUCHB.

As for the soundness, we want to decode an UG assignment from $T \subseteq V_H$ of size at most $|V_H|/2$ which contains at least $\varepsilon$ fraction of hyperedges. In terms of the tests, this corresponds to a collection of functions $\{f_u\}_{u \in \cV}$ such that $\E_{u \sim \cV} \E_{x \sim \{0, 1\}^R} f_u(x) \leqs 1/2$ and the Bansal-Khot test on $f_{v_1}, \dots, f_{v_t}$ passes with probability at least $\varepsilon$ where $v_1, \dots, v_t$ are sampled as in Figure~\ref{fig:red1}. Now, if we assume that $\E_x f_u(x) \leqs 0.99$ for all $u \in \cV$, then such decoding is possible via a similar method as in~\cite{BK09} since Theorem~\ref{thm:aint-over} can be applied here. 

Unfortunately, the assumption $\E_x f_u(x) \leqs 0.99$ does not hold for an arbitrary $T \subseteq V_H$ and the soundness property indeed fails. For instance, imagine the constraint graph $\cG$ of the starting UG instance consisting of two disconnected components of equal size; let $\cV_0, \cV_1$ be the set of vertices in the two components. In this case, the bisection $(\cV_0 \times \{0, 1\}^R, \cV_1 \times \{0, 1\}^R)$ does not even cut a single edge! This is regardless of whether there exists an assignment to the UG that satisfies a large fraction of edges.

\subsection{RST Technique and The Reduction from SSE to MUCHB} \label{subsec:red-final}

The issue described above is common for graph problems that involves some form of expansion of the graph. The RST technique~\cite{RST12} was in fact invented to specifically circumvent this issue. It works by first reducing SSE to UG and then exploiting the structure of the constructed UG instance when composing it with a long code test; this allows them to avoid extreme cases such as one above. There are four parameters in the reduction: $R, k \in \mathbb{N}$ and $\varepsilon_V, \beta$. Before we describe the reduction, let us define additional notations:
\begin{itemize}
\item Let $G^{\otimes R}$ denote the $R$-tensor graph of $G = (V, E, w)$; the vertex set of $G^{\otimes R}$ is $V^R$ and, for every $A, B \in V^R$, the edge weight between $A, B$ is the product of $w(A_i, B_i)$ in $G$ for all $i \in [R]$.
\item For each $A \in V^R$, $T_V(A)$ denote the distribution on $V^R$ where the $i$-th coordinate is set to $A_i$ with probability $1 - \varepsilon_V$ and is randomly sampled from $V$ otherwise.
\item Let $\Pi_{R, k}$ denote the set of all permutations $\pi$'s of $[R]$ such that, for each $j \in [k]$, $\pi(\{R(j - 1)/k + 1, \dots, Rj/k\}) = \{R(j - 1)/k + 1, \dots, Rj/k\}$.
\item Let $\{0, 1, \bot\}_\beta$ denote the probability space such that the probability for $0, 1$ are both $\beta/2$ and the probability for $\bot$ is $1 - \beta$.
\end{itemize}

The first step of reduction takes an SSE($\delta, \eta, M$) instance $G = (V, E, w)$ and produces a unique game $\cU = (\cG = (\cV, \cE, \cW), [R], \{\pi_e\}_{e \in E})$ where $\cV = V^{R}$ and the edges are distributed as follows:
\begin{enumerate}
\item Sample $A \sim V^R$ and $\tA \sim T_V(A)$.
\item Sample $B \sim G^{\otimes R}(\tA)$ and $\tB \sim T_V(B)$.
\item Sample two random $\pi_A, \pi_B \sim \Pi_{R, k}$.
\item Output an edge $e = (\pi_A(\tA), \pi_B(\tB))$ with $\pi_{(A, B)} = \pi_B \circ \pi^{-1}_A$.
\end{enumerate}

Here $\varepsilon_V$ is a small constant, $k$ is large and $R/k$ should be think of as $\Theta(1/\delta)$. When there exists a set $S \subseteq V$ of size $\delta|V|$ with small edge expansion, the intended assignment is to, for each $A \in V^R$, find the first block $j \in [k]$ such that $|A(j) \cap S| = 1$ where $A(j)$ denotes the multiset $\{A_{R(j-1)/k + 1}, \dots, A_{Rj/k}\}$ and let $F(A)$ be the coordinate of the vertex in that intersection. If no such $j$ exists, we assign $F(A)$ arbitrarily. Note that, since $R/k = \Theta(1/\delta)$, $\Pr[|A(j) \cap S| = 1]$ is constant, which means that only $2^{-\Omega(k)}$ fraction of vertices are assigned arbitrarily. Moreover, it is not hard to see that, for the other vertices, their assignments rarely violate constraints as $\varepsilon_V$ and $\Phi(S)$ are small. This yields the completeness. In addition, the soundness was shown in~\cite{RS10, RST12}, i.e., if every $S \subseteq V$ of size $\delta|V|$ has near perfect expansion, no assignment satisfies many constraints in $\cU$ (see Lemma~\ref{lem:decoding-sse}).

The second step is to reduce this UG instance to a hypergraph $H = (V_H, E_H)$. Instead of making the vertex set $V^R \times \{0, 1\}^R$ as in the previous candidate reduction, we will instead make $V_H = V^R \times \Omega^R$ where $\Omega = \{0, 1, \bot\}_\beta$ and $\beta$ is a small constant. This does not seem to make much sense from the UG reduction standpoint because we typically want to assign which side of the bisection $(A, x) \in V_H$ is in according to $x_{F(A)}$ but $x_{F(A)}$ could be $\bot$ in this construction. However, it makes sense when we view this as a reduction from SSE directly: let us discard all coordinates $i$'s such that $x_i = \bot$ and define $A(j, x) \triangleq \{A_i \mid i \in \{R(j - 1)/k + 1, \dots, Rj/k\} \wedge x_i \ne \bot\}$. Then, let $j^*(A, x) \triangleq \min\{j \mid |A(j, x) \cap S| = 1\}$ and let $i^*(A, x)$ be the coordinate in the intersection between $A(j^*(A, x), x)$ and $S$, and assign $(A, x)$ to $T_{x_{i^*(A, x)}}$. (If $j^*(A, x)$ does not exists, then assign $(A, x)$ arbitrarily.)

Observe that, in the intended solution, the side that $(A, x)$ is assigned to does not change if (1) $A_i$ is modified for some $i \in [R]$ s.t. $x_i = \bot$ or (2) we apply some permutation $\pi \in \Pi_{R, k}$ to both $A$ and $x$. In other words, we can ``merge'' two vertices $(A, x)$ and $(A', x')$ that are equivalent through these changes together in the reduction. For notational convenience, instead of merging vertices, we will just modify the reduction so that, if $(A, x)$ is included in some hyperedge, then every $(A', x')$ reachable from $(A, x)$ by these operations is also included in the hyperedge. More specifically, if we define $M_x(A) \triangleq \{A' \in V^R \mid A'_i = A_i \text{ for all } i \in [R] \text{ such that } x_i \ne \bot\}$ corresponding to the first operation, then we add $\pi(A', x)$ to the hyperedge for every $A' \in M_x(A)$ and $\pi \in \Pi_{R, k}$. The full reduction is shown in Figure~\ref{reduction}.

\begin{figure}[h!]
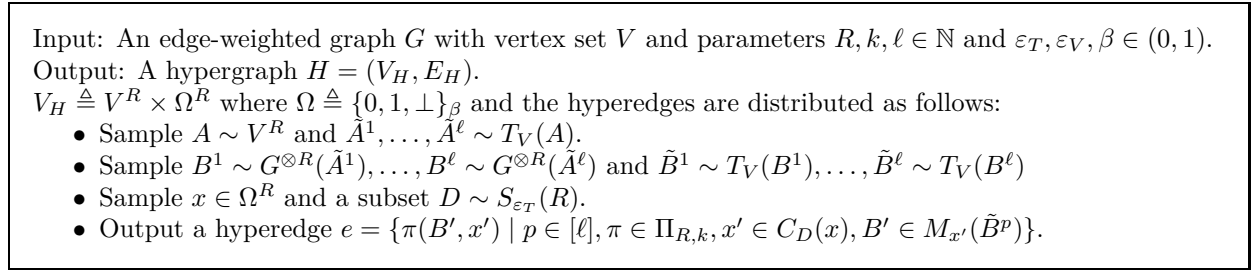

\begin{framed}
Input: An edge-weighted graph $G$ with vertex set $V$ and parameters $R, k, \ell \in \mathbb{N}$ and $\varepsilon_T, \varepsilon_V, \beta \in (0, 1)$. \\
Output: A hypergraph $H = (V_H, E_H)$. \\
$V_H \triangleq V^R \times \Omega^R$ where $\Omega \triangleq \{0, 1, \bot\}_\beta$ and the hyperedges are distributed as follows:
\begin{itemize}
\item Sample $A \sim V^R$ and $\tA^1, \dots, \tA^\ell \sim T_V(A)$.
\item Sample $B^1 \sim G^{\otimes R}(\tA^1), \dots, B^\ell \sim G^{\otimes R}(\tA^\ell)$ and $\tB^1 \sim T_V(B^1), \dots, \tB^\ell \sim T_V(B^\ell)$
\item Sample $x \in \Omega^R$ and a subset $D \sim S_{\varepsilon_T}(R)$.
\item Output a hyperedge $e = \{\pi(B', x') \mid p \in [\ell], \pi \in \Pi_{R, k}, x' \in C_D(x), B' \in M_{x'}(\tB^p)\}$.
\end{itemize}
\end{framed}
\caption{Reduction from SSE to Max UnCut Hypergraph Bisection}
\label{reduction}
\end{figure}

Note that the test we apply here is slightly different from Bansal-Khot test as our test is on $\Omega = \{0, 1, \bot\}_\beta$ instead of $\{0, 1\}$ used in~\cite{BK09}. Another thing to note is that now our vertices and hyperedges are weighted, the vertices according to the product measure of $V^R \times \Omega^R$ and the edges according to the distribution produced from the reduction. We write $\mu_{H}$ to denote the measure on the vertices, i.e., for $T \subseteq V^R \times \{0, 1, \bot\}^R$, $\mu_H(T) = \Pr_{A \sim V^R, x \sim \Omega^R}[(A, x) \in T]$, and we abuse the notation $E_H(T)$ and use it to denote the probability that a hyperedge as generated in Figure~\ref{reduction} lies completely in $T$. We note here that, while the MUCHB as stated in Lemma~\ref{lem:hypergraph-bisec} is unweighted, it is not hard to see that we can go from weighted version to unweighted by copying each vertex and each edge proportional\footnote{Note that this is doable since we can pick $\beta, \varepsilon_V, \varepsilon_T$ to be rational.} to their weights.

The advantage of this reduction is that the vertex ``merging'' makes gadget reduction non-local; for instance, it is clear that even if the starting graph $V$ has two connected components, the resulting hypergraph is now connected. In fact, Raghavendra \etal~\cite{RST12} show a much stronger quantitative bound. To state this, let us consider any $T \in V_H$ with $\mu_H(T) = 1/2$. From how the hyperedges are defined, we can assume w.l.o.g. that, if $(A, x) \in T$, then $\pi(A', x) \in T$ for every $A' \in M_x(A)$ and every $\pi \in \Pi_{R, k}$. Again, let $f_A(x) \triangleq \mathds{1}[(A, x) \in T]$. The following bound on the variance of $\E_x f_A(x)$ is implied by the proof of Lemma 6.6 in~\cite{RST12}:
\begin{align*}
\E_{A \sim V^R} \left(\E_{x \sim \Omega^R} f_A(x) - 1/2\right)^2 \leqs \beta.
\end{align*}

The above bound implies that, for most $A$'s, the mean of $f_A$ cannot be too large. This will indeed allow us to ultimately apply Theorem~\ref{thm:aint-over} on a certain fraction of the tuples $(\tB^{1}, \dots, \tB^{\ell})$ in the reduction, which leads to an UG assignment with non-negligible value.

\subsection{Completeness}

In the completeness case, we define a bisection similar to that described above. This bisection indeed cuts only a small fraction of hyperedges; quantitatively, this yields the following lemma.

\begin{lemma} \label{lem:completeness}
If there is a set $S \subseteq V$ such that $\Phi(S) \leqs \eta$ and $|S| = \delta |V|$ where $\delta \in \left[\frac{k}{10\beta R}, \frac{k}{\beta R}\right]$, then there is a bisection $(T_0, T_1)$ of $V_H$ such that $E_H(T_0), E_H(T_1) \geqs 1/2 - O(\varepsilon_T / \beta) - O(\eta \ell / \beta) - O(\varepsilon_V \ell / \beta) - 2^{-\Omega(k)}$ where $O(\cdot)$ and $\Omega(\cdot)$ hide only absolute constants.
\end{lemma}

\begin{proof}
Suppose that there exists $S \subseteq V$ of size $|S| = \delta |V|$  where $\left[\frac{k}{10\beta R}, \frac{k}{\beta R}\right]$ and $\Phi(S) \leqs \eta$. For $A \in V^R, x \in \{0, 1, \bot\}^R$, we will use the following notations throughout this proof:
\begin{itemize}
\item For $j \in [k]$, let $W(A, x, j)$ denote the set of all coordinates $i$ in $j$-th block such that $x_i \ne \bot$ and $A_i \in S$, i.e., $W(A, x, j) = \{i \in \{R(j - 1)/k + 1, \dots, Rj/k\} \mid A_i \in S \wedge x_i \ne \bot\}$.
\item Let $j^*(A, x)$ denote the first block $j$ with $|W(A, x, j)| = 1$, i.e., $j^*(A, x) = \min\{j \in [k] \mid |W(A, x, j)| = 1\}$. Note that if such block does not exist, we set $j^*(A, x) = -1$.
\item Let $i^*(A, x)$ be the only element in $W(A, x, j^*(A, x))$. If $j^*(A, x) = -1$, let $i^*(A, x) = -1$.
\end{itemize}

To define $T_0, T_1$, we start by constructing $T'_0 \subseteq T_0$ and $T'_1 \subseteq T_1$ as follows: assign each $(A, x) \in V_H$ such that $j^*(A, x) \ne -1$ to $T_{x_{i^*(A, x)}}$. Finally, we assign the rest of the vertices arbitrarily to $T_0$ and $T_1$ in such a way that $\mu_H(T_0) = \mu_H(T_1)$. Since $T'_0 \subseteq T_0, T'_1 \subseteq T_1$, it suffices to show the desired bound for $E_H(T'_0), E_H(T'_1)$. Due to symmetry, it suffices to bound $E_H(T'_0)$. Recall that $E_H(T'_0) = \Pr[e \subseteq T'_0]$ where $e$ is generated as detailed in Figure~\ref{reduction}.

To compute $E_H(T'_0)$, it will be most convenient to make a block-by-block analysis. In particular, for each block $j \in [k]$, we define $G_j$ to denote the event that $j^*(B', x') = j$ for some $(B', x') \in e$. We will be interested in bounding the following conditional probabilities:
\begin{itemize}
\item $c_1 \triangleq \Pr[j^*(A, x) = j \mid j^*(A, x) > j - 1]$
\item $c_2 \triangleq \Pr[G_j \mid j^*(A, x) > j \wedge \neg G_1 \wedge \cdots \wedge \neg G_{j - 1}]$
\item $c_3 \triangleq \Pr[e \nsubseteq T'_0 \mid j^*(A, x) = j \wedge \neg G_1 \wedge \cdots \wedge \neg G_{j - 1}]$
\end{itemize}
Here and throughout the proof, $e, A, \tA^1, \dots, \tA^\ell, B^1, \dots, B^\ell, \tB^1, \dots, \tB^\ell$ are as sampled by our reduction in Figure~\ref{reduction}. Note also that it is clear that $c_1, c_2, c_3$ do not depend on $j$.

Before we bound $c_1, c_2, c_3$, let us see how these probabilities can be used to bound $E_H(T'_0)$.
\begin{align}
\Pr_{e \sim E_H}[e \subseteq T'_0] \nonumber
&\geqs \sum_{j=1}^k \Pr_{e \sim E_H}[e \subseteq T'_0 \wedge j^*(A, x) = j] \nonumber \\
&\geqs \sum_{j=1}^k \Pr_{e \sim E_H}[e \subseteq T'_0 \wedge j^*(A, x) = j \wedge \neg G_1 \wedge \cdots \wedge \neg G_{j - 1}] \nonumber \\
&= \sum_{j=1}^k \Pr_{e \sim E_H}[e \subseteq T'_0 \mid j^*(A, x) = j \wedge \neg G_1 \wedge \cdots \wedge \neg G_{j - 1}]\Pr[j^*(A, x) = j \wedge \neg G_1 \wedge \cdots \wedge \neg G_{j - 1}] \nonumber \\
&= (1 - c_3)  \sum_{j=1}^k \Pr[j^*(A, x) = j \wedge \neg G_1 \wedge \cdots \wedge \neg G_{j - 1}] \nonumber \\
&= (1 - c_3)  \sum_{j=1}^k \Pr[j^*(A, x) = j]\Pr[\neg G_1 \wedge \cdots \wedge \neg G_{j - 1} \mid j^*(A, x) = j]. \label{eq:compl1}
\end{align}

The probability that $j^*(A, x) = j$ is in fact simply
\begin{align*}
\Pr[j^*(A, x) = j] = \Pr[j^*(A, x) = j \mid j^*(A, x) > j - 1] \prod_{q=1}^{j-1} \Pr[j^*(A, x) \ne q \mid j^*(A, x) > q - 1] = c_1(1 - c_1)^{j - 1}.
\end{align*}

Moreover, $\Pr[\neg G_1 \wedge \cdots \wedge \neg G_{j - 1} \mid j^*(A, x) = j]$ can be written as
\begin{align*}
\Pr[\neg G_1 \wedge \cdots \wedge \neg G_{j - 1} \mid j^*(A, x) = j] = \prod_{q=1}^{j-1} \Pr[\neg G_q \mid j^*(A, x) = j \wedge \neg G_1 \wedge \cdots \wedge \neg G_{q - 1}] = (1 - c_2)^{j - 1}.
\end{align*}

Plugging these two back into (\ref{eq:compl1}), we have
\begin{align}
E_H(T'_0) &\geqs (1 - c_3)\sum_{j=1}^k c_1(1 - c_1)^{j - 1}(1 - c_2)^{j - 1} \nonumber \\
&\geqs c_1(1 - c_3) \sum_{j=1}^k (1 - c_1 - c_2)^{j - 1} \nonumber \\
&= c_1(1 - c_3) \left(\frac{1 - (1 - c_1 - c_2)^k}{c_1 + c_2}\right) \nonumber \\
&= (1 - c_3)(1 - (1 - c_1 - c_2)^k)\left(\frac{1}{1 + c_2/c_1}\right) \nonumber \\
&\geqs (1 - c_3)(1 - (1 - c_1)^k)(1 - c_2/c_1) \nonumber \\
&\geqs 1 - c_3 - (1 - c_1)^k - c_2/c_1. \label{eq:compl-bound}
\end{align}

With (\ref{eq:compl-bound}) in mind, we will proceed to bound $c_1, c_2, c_3$. Before we do so, let us state two inequalities that will be useful: for every $i \in [R], p \in [\ell]$, we have
\begin{align} \label{eq:in-out}
\Pr[\tB^p_i \in S \mid A_i \notin S] \leqs 2\varepsilon_V \delta + 2\eta \delta
\end{align}
and 
\begin{align} \label{eq:out-in}
\Pr[\tB^p_i \notin S \mid A_i \in S] \leqs 2\varepsilon_V + \eta.
\end{align}
The first inequality comes from the fact that, for $\tB^p_i$ to be in $S$ when $A_i \notin S$, at least one of the following events must occur: (1) $A_i \ne \tA^p_i$ and $\tA^p_i \in S$, (2) $B^p_i \ne \tB^p_i$ and $\tB^p_i \in S$, (3) $(\tA^p_i, \tB^p_i) \in (V \setminus S) \times S$. Each of first two occurs with probability $\varepsilon_V \delta$ whereas the last event occurs with probability at most $\eta \delta / (1 - \delta) \leqs 2 \eta \delta$. On the other hand, for the second inequality, at least one of the following events must occur: (1) $A_i \ne \tA^p_i$, (2) $B^p_i \ne \tB^p_i$, (3) $(\tA^p_i, \tB^p_i) \in S \times (V \setminus S)$. Each of first two occurs with probability $\varepsilon_V$ whereas the last event occurs with probability at most $\eta$.

\subsubsection*{Bounding $c_1$}
To compute $c_1$, observe that $\Pr[j^*(A, z) = j \mid j^*(A, z) > j - 1]$ is the probability that, for exactly one $i$ in the $j$-th block, $A_i \in S$ and $x_i \ne \bot$. For a fixed $i$, this happens with probability $\beta \delta$. Hence, $c_1 = (R / k)\beta \delta(1 - \beta \delta)^{R/k - 1}$. Since $\delta \in \left[\frac{k}{10\beta R}, \frac{k}{\beta R}\right]$,
 we can conclude that $c_1$ is simply a constant (i.e. $c_1 \in [10^{-5}, 0.5]$).

\subsubsection*{Bounding $c_2$}
We next bound $c_2$. If $j^*(A, x) > j$, we know that $|W(A, x, j)| \ne 1$. Let us consider two cases:
\begin{enumerate}
\item $W(A, x, j) = \emptyset$. Observe that, if $G_j$ occurs, then there exist $p \in [\ell]$ and $i \in \{R(j - 1)/k + 1, \dots, Rj/k\}$ such that $\tB^p_i \in S$, and $x_i \ne \bot$ or $i \in D$. For brevity, below we denote the conditional event $j^*(A, x) > j \wedge \neg G_1 \wedge \cdots \wedge \neg G_{j - 1} \wedge W(A, x, j) = \emptyset$ by $E$. By union bound, our observation gives the following bound.
\begin{align}
\Pr[G_j \mid E] 
&\leqs \sum_{i=R(j - 1)/k+1}^{Rj/k} \Pr[\exists p \in [\ell], \tB^p_i \in S \wedge (x_i \ne \bot \vee i \in D) \mid E] \nonumber \\
&= \sum_{i=R(j-1)/k+1}^{Rj/k} \left(\Pr[\exists p \in [\ell], \tB^p_i \in S \wedge x_i \ne \bot \mid E] +  \Pr[\exists p \in [\ell], \tB^p_i \in S \wedge x_i = \bot \wedge i \in D \mid E]\right) \label{eq:compl-emptycase}
\end{align}


We can now bound the first term by
\begin{align}
\Pr[\exists p \in [\ell], \tB^p_i \in S \wedge x_i \ne \bot \mid E] &\leqs \Pr[\exists p \in [\ell], \tB^p_i \in S \mid x_i \ne \bot \wedge E] \nonumber \\
&\leqs \sum_{p=1}^\ell \Pr[\tB^p_i \in S \mid x_i \ne \bot \wedge E] \nonumber \\
(\text{Since } W(A, j) = \emptyset)&= \sum_{p=1}^\ell \Pr[\tB^p_i \in S \mid A_i \notin S] \nonumber \\
(\text{From }(\ref{eq:in-out})) &\leqs \ell(2\varepsilon_V \delta + 2\eta \delta). \label{eq:compl-tmp1}
\end{align}

Consider the other term in (\ref{eq:compl-emptycase}). We can rearrange it as follows.
\begin{align}
&\Pr[\exists p \in [\ell], \tB^p_i \in S \wedge x_i = \bot \wedge i \in D \mid E] \nonumber \\
&= \varepsilon_T \Pr[\exists p \in [\ell], \tB^p_i \in S \wedge x_i = \bot \mid i \in D \wedge E] \nonumber \\
&\leqs \varepsilon_T \Pr[\exists p \in [\ell], \tB^p_i \in S \mid x_i = \bot \wedge i \in D \wedge E] \nonumber \\
&\leqs \varepsilon_T \left(\Pr[A_i \in S \mid x_i = \bot \wedge i \in D \wedge E] + \Pr[\exists p \in [\ell], \tB^p_i \in S \mid A_i \notin S \wedge x_i = \bot \wedge i \in D \wedge E]\right) \nonumber \\
&= \varepsilon_T \left(\Pr[A_i \in S \mid x_i = \bot] + \Pr[\exists p \in [\ell], \tB^p_i \in S \mid A_i \notin S]\right) \nonumber \\
(\text{From }(\ref{eq:in-out})) &\leqs \varepsilon_T \left(\delta + 2 \varepsilon_V \delta \ell + 2 \eta \delta \ell\right).
\label{eq:compl-tmp2}
\end{align}

Combining (\ref{eq:compl-emptycase}), (\ref{eq:compl-tmp1}) and (\ref{eq:compl-tmp2}) and from $\delta \leqs \frac{k}{\beta R}$, we have
\begin{align*}
\Pr[G_j \mid E] 
\leqs 1/(\beta \delta) \left(2\varepsilon_V \delta \ell + 2 \eta \delta \ell + \varepsilon_T \delta + 2 \varepsilon_V \delta \ell + 2 \eta \delta \ell \right)
\leqs O(\varepsilon_T / \beta) + O(\eta \ell / \beta) + O(\varepsilon_V \ell / \beta).
\end{align*}
\item $|W(A, x, j)| > 1$. Let $i^*_1$ and $i^*_2$ be two different (arbitrary) elements of $W(A, x, j)$. Again, for convenient, we use $E$ to denote the conditional event $j^*(A, x) > j \wedge \neg G_1 \wedge \cdots \wedge \neg G_{j - 1} \wedge \{i^*_1, i^*_2\} \subseteq W(A, x, j)$. Now, let us first split $\Pr[G_j | E]$ as follows.
\begin{align}
\Pr[G_j \mid E] &\leqs \Pr[i^*_1 \in D] + \Pr[i^*_2 \in D] + \Pr[G_j \mid E \wedge i^*_1, i^*_2 \notin D] \nonumber \\
&= 2\varepsilon_T + \Pr[G_j \mid E \wedge i^*_1, i^*_2 \notin D]. \label{eq:compl-tmp3}
\end{align}

Observe that, when $i^*_1, i^*_2 \notin D$, $x'_{i^*_1}, x'_{i^*_2} \ne \bot$ for every $x' \in C_D(x)$. Hence, for $G_j$ to occur, there must be $p \in [\ell]$ such that at least one of $\tB^p_{i^*_1}, \tB^p_{i^*_2}$ is not in $S$. In other words,
\begin{align}
\Pr[G_j \mid E \wedge i^*_1, i^*_2 \notin D] &\leqs \Pr[\exists p \in [\ell], \tB^p_{i^*_1} \notin S \vee \tB^p_{i^*_2} \notin S \mid E \wedge i^*_1, i^*_2 \notin D] \nonumber \\
&\leqs \sum_{p=1}^{\ell} \left(\Pr[\tB^p_{i^*_1} \notin S \mid E \wedge i^*_1, i^*_2 \notin D] + \Pr[\tB^p_{i^*_2} \notin S \mid E \wedge i^*_1, i^*_2 \notin D]\right) \nonumber \\
(\text{From } i^*_1, i^*_2 \in W(A, x, j)) &= \sum_{p=1}^{\ell} \left(\Pr[\tB^p_{i^*_1} \notin S \mid A_{i^*_1} \in S] + \Pr[\tB^p_{i^*_2} \notin S \mid A_{i^*_2} \in S]\right) \nonumber \\
(\text{From } (\ref{eq:out-in})) &\leqs 2\varepsilon_V\ell + \eta\ell.
\end{align}

Combining this with (\ref{eq:compl-tmp3}), we have
$\Pr[G_j \mid E] \leqs O(\varepsilon_T) + O(\eta \ell) + O(\varepsilon_V \ell)$.
\end{enumerate}
As a result, we can conclude that $c_2$ is at most $O(\varepsilon_T / \beta) + O(\eta \ell / \beta) + O(\varepsilon_V \ell / \beta)$.

\subsubsection*{Bounding $c_3$}
Finally, let us bound $c_3$. First, note that the probability that $x_{i^*(A, x)} = 1$ is $1/2$ and that the probability that $i^*(A, x) \in D$ is $\varepsilon_T$. This means that
\begin{align*}
c_3 &\leqs 1/2 + \varepsilon_T + \Pr[e \nsubseteq T'_0 \mid E].
\end{align*}
where $E$ is the event $x_{i^*(A, x)} = 0 \wedge i^*(A, x) \notin D \wedge j^*(A, x) = j \wedge \neg G_1 \wedge \cdots \wedge \neg G_{j - 1}$.

Moreover, since $A_{i^*(A, x)} \in S$, from (\ref{eq:out-in}) and from union bound, we have
\begin{align*}
\Pr\left[\exists p \in [\ell], \tB^p_{i^*(A, x)} \notin S \midv E\right] \leqs 2\varepsilon_V\ell + \eta\ell.
\end{align*}
From the above two inequalities, we have
\begin{align*}
c_3 &\leqs 1/2 + \varepsilon_T + 2\varepsilon_V\ell + \eta\ell + \Pr\left[e \nsubseteq T'_0 \midv E \wedge \left(\forall p \in [\ell], \tB^p_{i^*(A, x)} \in S\right)\right].
\end{align*}

Conditioned on the above event, $e \nsubseteq T'_0$ implies that there exists $p \in [\ell]$ and some $i \ne i^*(A, x)$ in this ($j$-th) block such that $\tB^p_i \in S$, and $x_i \ne \bot$ or $i \in D$. We have bounded an almost identical probability before in the case $W(A, x, j) = \emptyset$ when we bound $c_2$. Similary, here we have an upper bound of $O(\varepsilon_T / \beta) + O(\eta \ell / \beta) + O(\varepsilon_V \ell / \beta)$ on this probability. Hence,
\begin{align*}
c_3 \leqs 1/2 + O(\varepsilon_T / \beta) + O(\eta \ell / \beta) + O(\varepsilon_V \ell / \beta)
\end{align*}

By combining our bounds on $c_1, c_2, c_3$ with (\ref{eq:compl-bound}), we immediately arrive at the desired bound:
\begin{align*}
E_H(T'_0) \geqs 1/2 - O(\varepsilon_T / \beta) - O(\eta \ell / \beta) - O(\varepsilon_V \ell / \beta) - 2^{-\Omega(k)}.
\end{align*}
\end{proof}

\subsection{Soundness}

Let us consider any set $T$ such that $\mu_H(T) \leqs 1/2$. We would like to give an upper bound on $E_H(T)$. From how we define hyperedges, we can assume w.l.o.g. that $(A, x) \in T$ if and only if $\pi(A', x) \in T$ for every $A' \in M_x(A)$ and $\pi \in \Pi_{R, k}$.
We call such $T$ \emph{$\Pi_{R, k}$-invariant}.

Let $f: V^R \times \Omega^R \rightarrow \{0, 1\}$ denote the indicator function for $T$, i.e., $f(A, x) = 1$ if and only if $(A, x) \in T$. Note that $\E_{A \sim V^R, x \sim \Omega^R} f(A, x) = \mu_H(T) \leqs 1/2$.
Following notation from~\cite{RST12}, we write $f_A(x)$ as a shorthand for $f(A, x)$. In addition, for each $A \in V^R$, we will write $\tB \sim \Gamma(A)$ as a shorthand for $\tB$ generated randomly by sampling $\tA \sim T_V(A)$, $B \sim G^{\otimes R}(\tA)$ and $\tB \sim T_V(B)$ respectively. Let us restate Raghavendra \etal's~\cite{RST12} lemma regarding the variance of $\E_x f_A(x)$ in a more convenient formulation below.

\begin{lemma}[{\cite[Lemma~6.6]{RST12}}\footnotemark] \label{lem:mean-conc}
For every $A \in V^R$, let $\mu_A \triangleq \E_{x \sim \Omega^R} f_A(x)$. We have
\begin{align*}
\E_{A \sim V^R} \left(\E_{\tB \sim \Gamma(A)} \mu_{\tB} - \mu_H(T)\right)^2 \leqs \beta.
\end{align*}
\end{lemma}
\footnotetext{Lemma 6.6 in~\cite{RST12} involves symmetrizing $f$'s, but we do not need it here since $T$ is $\Pi_{R, k}$-invariant.}

To see how the above lemma helps us decode an UG assignment, observe that, if our test accepts on $f_{\tB^1}, \dots, f_{\tB^\ell}, x, D$, then it also accepts on any subset of the functions (with the same $x, D$); hence, to apply Theorem~\ref{thm:aint-over}, it suffices that $t$ of the functions have means $\leqs 0.99$. We will choose $\ell$ to be large compared to $t$. Using above lemma and a standard tail bound, we can argue that Theorem~\ref{thm:aint-over} is applicable for almost all tuples $\tB^1, \dots, \tB^\ell$, as stated below.

\begin{lemma} \label{lem:mean-test}
For any positive integer $t \leqs 0.01\ell$,
\begin{align*}
\Pr_{A \sim V^R, \tB^1, \dots, \tB^\ell \sim \Gamma(A)} \left[\left|\left\{i \in [\ell] \midv \mu_{\tB^i} \leqs 0.99\right\}\right| \geqs t\right] \geqs 1 - 10 \beta - 2^{-\ell / 100}.
\end{align*}
\end{lemma}

\begin{proof}
First, note that, since $\mu_H(T) \leqs 1/2$, we can use Cherbychev's inequality and Lemma 6.6 to arrive at the following bound, which is analogous to Lemma 6.7 in~\cite{RST12}:
\begin{align}
\Pr_{A \sim V^R}\left[\E_{\tB \sim \Gamma(A)} \mu_{\tB} \geqs 0.9\right] \leqs 10\beta.
\label{eq:badA}
\end{align}

Let us call $A \in V^R$ such that $\E_{\tB \sim \Gamma(A)} \mu_{\tB} \geqs 0.9$ \emph{bad} and the rest of $A \in V^R$ \emph{good}.

For any good $A \in V^R$, Markov's inequality implies that $\Pr_{\tB \sim \Gamma(A)}[\mu_{\tB} > 0.99] \leqs 0.9/0.99 < 0.95$. As a result, an application of Chernoff bound gives the following inequality.
\begin{align}
\Pr_{\tB^1, \dots, \tB^\ell \sim \Gamma(A)} \left[\left|\left\{i \in [\ell] \midv \mu_{\tB^i} \leqs 0.99\right\}\right| < t \midv A \text{ is good}\right]
\leqs 2^{-\ell/100}.
\label{eq:goodA}
\end{align}

Finally, observe that (\ref{eq:badA}) and (\ref{eq:goodA}) immediately yields the desired bound.
\end{proof}

\subsubsection{Decoding an Unique Games Assignment}

With Lemma~\ref{lem:mean-test} ready, we can now decode an UG assignment via a similar technique from~\cite{BK09}.

\begin{lemma} \label{lem:decode-assignment}
For any $\varepsilon_T, \gamma, \beta > 0$, let $t = t(\varepsilon_T, \gamma, \beta), \kappa = \kappa(\varepsilon_T, \gamma, \beta)$ and $d = d(\varepsilon_T, \gamma, \beta)$ be as in Theorem~\ref{thm:aint-over}. For any integer $\ell \geqs 100t$, if there exists $T \subseteq V_H$ of such that $\mu_H(T) \leqs 1/2$ and $E_H(T) \geqs 2\gamma + 10\beta + 2^{-\ell/100}$, then there exists $F: V^R \rightarrow [R]$ such that
\begin{align*}
\Pr_{A \sim V^R, \tB \sim \Gamma(A), \pi_A, \pi_B \sim \Pi_{R, k}}[\pi_A^{-1}(F(\pi_A(\tA))) = \pi_B^{-1}(F(\pi_B(\tB)))] \geqs \frac{\gamma \kappa^2}{4 d^2 \ell^2}.
\end{align*}
\end{lemma}

\begin{proof}
The decoding procedure is as follows. For each $A \in V^R$, we construct a set of candidate labels $\cand[A] \triangleq \{j \in [R] \mid \infl^{d}_j(f_A) \geqs \kappa\}$. We generate $F$ randomly by, with probability 1/2, setting $F(A)$ to be a random element of $\cand[A]$ and, with probability 1/2, sampling $\tB \sim \Gamma(A)$ and setting $F(A)$ to be a random element from $\cand[B]$. Note that, if the candidate set is empty, then we simply pick an arbitrary assignment.

From our assumption that $T$ is $\Pi_{R, k}$-invariant, it follows that, for every $A \in V^R, \pi \in \Pi_{R, k}$ and $j \in [R]$, $\Pr_F[\pi^{-1}(F(\pi(A))) = j] = \Pr_F[F(A) = j]$. In other words, we have
\begin{align}
\Pr_{F, A \sim V^R, \tB \sim \Gamma(A), \pi_A, \pi_B \sim \Pi_{R, k}}[\pi_A^{-1}(F(\pi_A(\tA))) = \pi_B^{-1}(F(\pi_B(\tB)))] &= \Pr_{F, A \sim V^R, \tB \sim \Gamma(A)}[F(\tA) = F(\tB)].
\label{eq:lem-f-inv}
\end{align}

Next, note that, from how our reduction is defined, $E_H(T)$ can be written as
\begin{align*}
E_H(T) = \Pr_{A \sim V^R, \tB^1, \dots, \tB^\ell \sim \Gamma(A), x \sim \Omega^R, D \sim S_{\varepsilon_T}(R)}\left[\bigwedge_{i=1}^\ell f_{\tB^i}(C_D(x)) \equiv 1\right].
\end{align*}

From $E_H(T) \geqs 2\gamma + 10\beta + 2^{-\ell/100}$ and from Lemma~\ref{lem:mean-test}, we can conclude that
\begin{align*}
\Pr_{A, \tB^1, \dots, \tB^\ell, x, D}\left[\left(\bigwedge_{i=1}^\ell f_{\tB^i}(C_D(x)) \equiv 1\right) \wedge \Big(\left|\left\{i \in [\ell] \midv \mu_{\tB^i} \leqs 0.99\right\}\right| \geqs t\Big)\right] \geqs 2\gamma.
\end{align*}

From Markov's inequality, we have
\begin{align*}
\gamma &\leqs \Pr_{A, \tB^1, \dots, \tB^\ell} \left[\Pr_{x, D}\left[\left(\bigwedge_{i=1}^\ell f_{\tB^i}(C_D(x)) \equiv 1\right) \wedge \Big(\left|\left\{i \in [\ell] \midv \mu_{\tB^i} \leqs 0.99\right\}\right| \geqs t\Big) \right] \geqs \gamma\right] \\
&= \Pr_{A, \tB^1, \dots, \tB^\ell} \left[\left(\Pr_{x, D}\left[\bigwedge_{i=1}^\ell f_{\tB^i}(C_D(x)) \equiv 1\right] \geqs \gamma\right) \wedge \Big(\left|\left\{i \in [\ell] \midv \mu_{\tB^i} \leqs 0.99\right\}\right| \geqs t\Big)\right].
\end{align*}

A tuple $(A, \tB^1, \dots, \tB^\ell)$ is said to be \emph{good} if $\Pr_{x \sim \Omega^R, D \sim S_{\varepsilon_T}(R)}\left[\bigwedge_{i=1}^\ell f_{\tB^i}(C_D(x)) \equiv 1\right] \geqs \gamma$ and $\left|\left\{i \in [\ell] \midv \mu_{\tB^i} \leqs 0.99\right\}\right| \geqs t$. For such tuple, Theorem~\ref{thm:aint-over} implies that there exist $i_1 \ne i_2 \in [\ell], j \in [R]$ s.t. $\infl_j^d(f_{\tB^{i_1}}), \infl_j^d(f_{\tB^{i_2}}) \geqs \kappa$. This means that $\cand(\tB^{i_1}) \cap \cand(\tB^{i_2}) \ne \emptyset$.

Hence, if we sample a tuple $(A, \tB^1, \dots, \tB^\ell)$ at random, and then sample two different $\tB, \tB'$ randomly from $\tB^1, \dots, \tB^\ell$, then the tuple is good with probability at least $\gamma$ and, with probability $1/\ell^2$, we have $\tB = \tB^{i_1}, \tB' = \tB^{i_2}$. This gives the following bound:
\begin{align*}
\Pr_{A, \tB, \tB'}\left[\cand(\tB) \cap \cand(\tB') \ne \emptyset\right] \geqs \frac{\gamma}{\ell^2}.
\end{align*}

Now, observe that $\tB$ and $\tB'$ above are distributed in the same way as if we pick both of them independently with respect to $\Gamma(A)$. Recall that, with probability 1/2, $F(A)$ is a random element of $\cand(\tB)$ where $\tB \sim \Gamma(A)$ and, with probability 1/2, $F(\tB')$ is a random element of $\cand(\tB')$. Moreover, since the sum of degree $d$-influence is at most $d$~\cite[Proposition 3.8]{MOO10}, the candidate sets are of sizes at most $d/\kappa$. As a result, the above bound yields
\begin{align*}
\Pr_{A \sim V^R, \tB' \sim \Gamma(A)}[F(A) = F(\tB')] \geqs \frac{\gamma \kappa^2}{4 d^2 \ell^2},
\end{align*}
which, together with (\ref{eq:lem-f-inv}), concludes the proof of the lemma.
\end{proof}

\subsubsection{Decoding a Small Non-Expanding Set}

To relate our decoded UG assignment back to a small non-expanding set in $G$, we use the following lemma of~\cite{RST12}, which roughly states that, with the right parameters, the soundness case of SSEH implies that only small fraction of constriants in the UG can be satisfied.

\begin{lemma}[{\cite[Lemma~6.11]{RST12}}] \label{lem:decoding-sse}
If there exists $F: V^R \rightarrow [R]$ such that
\begin{align*}
\Pr_{A \sim V^R, \tB \sim \Gamma(A), \pi_A, \pi_B \sim \Pi_{R, k}}[\pi_A^{-1}(F(\pi_A(\tA))) = \pi_B^{-1}(F(\pi_B(\tB)))] \geqs \zeta,
\end{align*}
then there exists a set $S \subseteq V$ with $\frac{|S|}{|V|} \in \left[\frac{\zeta}{16R}, \frac{3k}{\varepsilon_V R}\right]$ with $\Phi(S) \leqs 1 - \frac{\zeta}{16k}$.
\end{lemma}

By combining the above lemma with Lemma~\ref{lem:decode-assignment}, we immediately arrive at the following:

\begin{lemma} \label{lem:soundness}
For any $\varepsilon_T, \gamma, \beta > 0$, let $t = t(\varepsilon_T, \gamma, \beta), \kappa = \kappa(\varepsilon_T, \gamma, \beta)$ and $d = d(\varepsilon_T, \gamma, \beta)$ be as in Theorem~\ref{thm:aint-over}. For any integer $\ell \geqs 100t$ and any $\varepsilon_V > 0$, if there exists $T \subseteq V_H$ with $\mu_H(T) \leqs 1/2$ such that $E_H(T) \geqs 2\gamma + 10\beta + 2^{-\ell/100}$, then there exists a set $S \subseteq V$ with $\frac{|S|}{|V|} \in \left[\frac{\zeta}{16R}, \frac{3k}{\varepsilon_V R}\right]$ with $\Phi(S) \leqs 1 - \frac{\zeta}{16k}$ where $\zeta = \frac{\gamma \kappa^2}{4 d^2 \ell^2}$.
\end{lemma}

\subsection{Putting Things Together}

We can now deduce inapproximability of MUCHB by simply picking appropriate parameters.

\begin{proof}[Proof of Lemma~\ref{lem:hypergraph-bisec}]
The parameters are chosen as follows:
\begin{itemize}
\item Let $\beta = \varepsilon / 30, \gamma = \varepsilon / 6$, and $k = \Omega(\log (1/\varepsilon))$ so that the term $2^{-\Omega(k)}$ in Lemma~\ref{lem:completeness} is $\leqs \varepsilon/4$.
\item Let $\varepsilon_T = O(\beta \varepsilon)$ so that the error term $O(\varepsilon_T/\beta)$ in Lemma~\ref{lem:completeness} is at most $\varepsilon/4$.
\item Let $t = t(\varepsilon_T, \gamma, \beta), \kappa = \kappa(\varepsilon_T, \gamma, \beta)$ and $d = d(\varepsilon_T, \gamma, \beta)$ be as in Theorem~\ref{thm:aint-over}.
\item Let $\zeta = \frac{\gamma \kappa^2}{4 d^2 \ell^2}$ be as in Lemma~\ref{lem:soundness} and let $\ell = \max\{100 t, 1000 \log(1/\varepsilon)\}$.
\item Let $\varepsilon_V = O(\varepsilon \beta / \ell)$ where so that the error term $O(\varepsilon_V\ell/\beta)$ in Lemma~\ref{lem:completeness} is at most $\varepsilon/4$.
\item Let $\eta = \min\{\frac{\zeta}{32k}, O(\varepsilon \beta/\ell)\}$ so that the error term $O(\eta\ell/\beta)$ in Lemma~\ref{lem:completeness} is at most $\varepsilon/4$.
\item Let $M = \max\{\frac{16k}{\beta\zeta}, \frac{3\beta}{\varepsilon_V}\}$.
\item Finally, let $R = \frac{k}{\beta \delta}$ where $\delta = \delta(\eta, M)$ is the parameter from the SSEH (Conjecture~\ref{conj:sse-strong}).
\end{itemize}

Let $G = (V, E, w)$ be an instance of SSE($\eta, \delta, M$) and let $H = (V_H, G_H)$ be the hypergraph resulted from our reduction. If there exists $S \subseteq V$ of size $\delta|V|$ of expansion at most $\eta$, Lemma~\ref{lem:completeness} implies that there is a bisection $(T_0, T_1)$ of $V_H$ such that $E_H(T_0), E_H(T_1) \geqs 1/2 - \varepsilon$.

As for the soundness, Lemma~\ref{lem:soundness} with our choice of parameters implies that, if there exists a set $T \subseteq V_H$ with $\mu(T) \leqs 1/2$ and $E_H(T_0) \geqs \varepsilon$, there exists $S \subseteq V$ with $|S| \in \left[\frac{\delta|V|}{M}, \delta|V|M\right]$ whose expansion is less than $1 - \eta$. The contrapositive of this yields the soundness property.
\end{proof}

\section{Conclusion}

In this work, we prove essentially tight inapproximability of MEB, MBB, Minimum $k$-Cut and DAL$k$S based on SSEH. Our results, expecially for the biclique problems, demonstrate further the applications of the hypothesis and particularly the RST technique~\cite{RST12} in proving hardness of graph problems that involve some form of expansion. Given that the technique has been employed for only a handful of problems~\cite{RST12,LRV13}, an obvious but intriguing research direction is to try to utilize the technique to other problems. One plausible candidate problem to this end is the 2-Catalog Segmentation Problem~\cite{KPR04} since a natural candidate reduction for this problem fails due to a similar counterexample as in Section~\ref{subsec:bk}. 

Another interesting question is to derandomize graph product used in the gap amplification step for biclique problems. For Maximum Clique, this step has been derandomized before~\cite{AFWZ95,Zuc07}; in particular, Zuckerman~\cite{Zuc07} derandomized H{\aa}stad's result~\cite{Has01} to achieve $n^{1 - \varepsilon}$ ratio NP-hardness for approximating Maximum Clique. Without going into too much detail, we would like to note that Zuckerman's result is based on a construction of dispersers with certain parameters; properties of dispersers then imply soundness of the reduction whereas completeness is trivial from the construction since H{\aa}stad's PCP has perfect completeness. Unfortunately, our PCP does not have perfect completeness and, in order to use Zuckerman's approach, additional properties are required in order to argue about completeness of the reduction.

\subsection*{Acknowledgement}

I am grateful to Prasad Raghavendra for providing his insights on the Small Set Expansion problem and techniques developed in~\cite{RS10, RST12} and Luca Trevisan for lending his expertise in PCPs with small free bits. I would also like to thank Haris Angelidakis for useful discussions on Minimum $k$-Cut and Daniel Reichman for inspiring me to work on the two biclique problems. Finally, I thank ICALP 2017 anonymous reviewers for their useful comments and, more specifically, for pointing me to~\cite{BGHKK16}.

\bibliographystyle{alpha}
\bibliography{main}

\appendix

\section{Reduction from MUCHB to Biclique Problems} \label{app:red-mchb}

\begin{proof}[Proof of Lemma~\ref{lem:meb-mbb-withoutamp} from Lemma~\ref{lem:hypergraph-bisec}]
The reduction from MUCHB to the two biclique problems are simple. Given a hypergraph $H = (V_H, E_H)$, we create a bipartite graph $G = (L, R, E)$ by letting $L = R = E_H$ and creating an edge $(e_1, e_2)$ for $e_1 \in L, e_2 \in R$ iff $e_1 \cap e_2 = \emptyset$.

{\bf Completeness.} If there is a bisection $(T_0, T_1)$ of $V_H$ such that $|E_H(T_0)|, |E_H(T_1)| \geqs (1/2 - \varepsilon)|E_H|$, then $E_H(T_0) \subseteq L$ and $E_H(T_1) \subseteq R$ induce a complete bipartite subgraph with at least $(1/2 - \varepsilon)n$ vertices on each side in $G$.

{\bf Soundness.} Suppose that, for every $T \subseteq V_H$ of size at most $|V_H|/2$, we have $|E_H(T)| \leqs \varepsilon |V_H|$. We will show that $G$ does not contain $K_{\varepsilon n + 1, \varepsilon n + 1}$ as a subgraph. Suppose for the sake of contradiction that $G$ contains $K_{\varepsilon n + 1, \varepsilon n + 1}$ as a subgraph. Consider one such subgraph; let $e_1, \dots, e_{\varepsilon n + 1} \in L$ and $e'_1, \dots, e'_{\varepsilon n + 1} \in R$ be the vertices of the subgraph. From how the edges are defined, $T_0 \triangleq e_1 \cup \cdots \cup e_{\varepsilon n + 1}$ and $T_1 \triangleq e'_1 \cup \cdots \cup e'_{\varepsilon n + 1}$ are disjoint. At least one of the two sets is of size at most $|V_H|/2$, assume without loss of generality that $|T_0| \leqs |V_H|/2$. This is a contradiction since $T_0$ contains at least $\varepsilon n + 1$ hyperedges $e_1, \dots, e_{\varepsilon n + 1}$.

By picking $\varepsilon = \delta/2$ and $n > 2 / \varepsilon$, we arrive at the desired hardness result.
\end{proof}

\section{Gap Amplification via Randomized Graph Product} \label{app:amp}

In this section, we provide a full proof of the gap amplification step for biclique problems, thereby proving Theorem~\ref{thm:meb-mbb}. The argument provided below is almost the same as the randomized graph product based analysis of Khot~\cite[Appendix D]{Kho06}, which is in turn based on the analysis of Berman and Schnitger~\cite{BS92}, except that we modify the construction slightly so that the reduction time is polynomial\footnote{Using Khot's result directly would result in the construction time being quasi-polynomial.}.

Specifically, we will prove the following statement, which immediately implies Theorem~\ref{thm:meb-mbb}.

\begin{lemma} \label{lem:meb-mbb-amp}
Assuming SSEH and NP $\ne$ BPP, for every $\varepsilon > 0$, no polynomial time algorithm can, given a bipartite graph $G' = (L', R', E')$ with $|L'| = |R'| = N$, distinguish between the following two cases:
\begin{itemize}
\item \emph{(Completeness)} $G'$ contains $K_{N^{1 - \varepsilon}, N^{1 - \varepsilon}}$ as a subgraph.
\item \emph{(Soundness)} $G'$ does not contain $K_{N^{\varepsilon}, N^{\varepsilon}}$ as a subgraph.
\end{itemize}
\end{lemma}

\begin{proof}
Let $\varepsilon > 0$ be any constant. Assume without loss of generality that $\varepsilon < 1$. Consider any bipartite graph $G = (L, R, E)$ with $|L| = |R| = n$. Let $\delta = 2^{-4/\varepsilon}, k = \log n$ and $N = (1/\delta)^k$. We construct a graph $G' = (L', R', E')$ where $|L'| = |R'| = N$ as follows. For $i \in [N]$, pick random elements $U^i \sim L^k$ and $V^i \sim R^k$, and add them to $L'$ and $R'$ respectively. Finally, for every $U \in L'$ and $V \in R'$, there is an edge between $U$ and $V$ in $G'$ if and only if there is an edge in the original graph $G$ between $U_{j_1}$ and $V_{j_2}$ for every $j_1, j_2 \in [k]$.

{\bf Completeness.} Suppose that the original graph $G$ contains $K_{(1/2 - \delta)n, (1/2 - \delta)n}$ as a subgraph. Let one such subgraph be $(S, T)$ where $S \subseteq L$ and $T \subseteq R$. Observe that $(L' \cap S^k, R' \cap T^k)$ induces a biclique in $G'$. For each $i \in [N]$, note that $U_i \in S^k$ with probability $(1/2 - \delta)^k \geqs (1/4)^k = N^{-\varepsilon/2}$ independent of each other. As a result, from Chernoff bound, $|L' \cap S^k| \geqs N^{1 - \varepsilon}$ with high probability. Similarly, we also have $|R' \cap T^k| \geqs N^{1 - \varepsilon}$ with high probability. Thus, $G'$ contains $K_{N^{1 - \varepsilon}, N^{1 - \varepsilon}}$ as a subgraph with high probability.

{\bf Soundness.} Suppose that $G$ does not contain $K_{\delta n, \delta n}$ as a subgraph. We will show that, with high probability, $G'$ does not contain $K_{N^\varepsilon, N^\varepsilon}$ as a subgraph. To do so, we will first prove the following proposition.

\begin{proposition} \label{prop:amp-aux}
For any set $A$, let $\cP(A)$ denote the power set of $A$. Moreover, let $\cF: \cP(L^k \cup R^k) \rightarrow \cP(L \cup R)$ be the ``flattening'' operation defined by $\cF(A) \triangleq \cup_{U \in A} \{U_i \mid i \in [k]\}$. Then, with high probability, we have $|\cF(S')| \geqs \delta n$ for every subset $S' \subseteq L'$ of size $N^\varepsilon$ and $|\cF(T')| \geqs \delta n$ for every subset $T' \subseteq R'$ of size $N^\varepsilon$.
\end{proposition}

\begin{proof}[Proof of Proposition~\ref{prop:amp-aux}]
Let us consider the probability that there exists a set $S' \subseteq L'$ of size $N^\varepsilon$ such that $|\cF(S')| < \delta n$. This can be bounded as follows.
\begin{align*}
\Pr[\exists S' \subseteq L', |S'| = N^\varepsilon \wedge \cF(S') < \delta n] &= \Pr[\exists S \subseteq L, |S| < \delta n \wedge |S^k \cap L'| \geqs N^\varepsilon] \\
&\leqs \sum_{S \subseteq L, |S| < \delta n} \Pr[|S^k \cap L'| \geqs N^\varepsilon] \\
&= \sum_{S \subseteq L, |S| < \delta n} \Pr\left[\sum_{i \in [N]} \ind[U_i \in S^k] \geqs N^\varepsilon\right].
\end{align*}

Observe that, for each $i \in [N]$, $\ind[U_i \in S^k]$ is simply an independent Bernoulli random variable with mean $(|S|/n)^k < \delta^k = 1/N$. Hence, by Chernoff bound, we have
\begin{align*}
\Pr[\exists S' \subseteq L', |S'| = N^\varepsilon \wedge \cF(S') < \delta n] 
\leqs \sum_{S \subseteq L, |S| < \delta n} 2^{-\Omega(N^\varepsilon)} \leqs \sum_{S \subseteq L, |S| < \delta n} 2^{-\Omega(n^2)} = 2^{-\Omega(n^2)}
\end{align*}
as desired.

Analogously, we also have $|\cF(T')| \geqs \delta n$ for every subset $T' \subseteq R'$ of size $N^\varepsilon$ with high probability, thereby concluding the proof of Proposition~\ref{prop:amp-aux}.
\end{proof}

With Proposition~\ref{prop:amp-aux} ready, let us proceed with our soundness proof. Suppose that the event in Proposition~\ref{prop:amp-aux} occurs. Consider any subset $S' \subseteq L'$ of size $N^\varepsilon$ and any subset $T' \subseteq R'$ of size $N^\varepsilon$. Since $|\cF(S')| \geqs \delta n, |\cF(T')| \geqs \delta n$ and $G$ does not contain $K_{\delta n, \delta n}$ as a subgraph, there exists $u \in \cF(S')$ and $v \in \cF(T')$ such that $(u, v) \notin E$. From the definition of $G'$, this implies that $S'$ and $T'$ do not induce a biclique in $G'$. As a result, $G'$ does not contain $K_{N^\varepsilon, N^\varepsilon}$ as a subgraph. From this and from Proposition~\ref{prop:amp-aux}, $G'$ does not contain $K_{N^\varepsilon, N^\varepsilon}$ as a subgraph with high probability, concluding our soundness argument.

Since Lemma~\ref{lem:meb-mbb-withoutamp} asserts that distinguishing between the two cases above are NP-hard (assuming SSEH) and the above reduction takes polynomial time, we can conclude that, assuming SSEH and NP $\ne$ BPP, no polynomial time algorithm can distinguish the two cases stated in the lemma.
\end{proof}

\section{Comparison Between SSEH and Strong UGC} \label{app:ugc}

In this section, we briefly discuss the similarities and differences between the classical Unique Games Conjecture~\cite{Kho02}, the Small Set Expansion Hypothesis~\cite{RS10} and the Strong Unique Games Conjecture~\cite{BK09}. Let us start by stating the Unique Games Conjecture, proposed by Khot in his influential work~\cite{Kho02}:

\begin{conjecture}[Unique Games Conjecture (UGC)~\cite{Kho02}]
For every $\varepsilon, \eta > 0$, there exists $R = R(\varepsilon, \eta)$ such that, given an UG instance $(\cG = (\cV, \cE, \cW), [R], \{\pi_e\}_{e \in \cE})$ such that $\cG$ is regular, it is NP-hard to distinguish between the following two cases:
\begin{itemize}
\item (Completeness) There exists an assignment $F: \cV \to [R]$ such that $\val_{\cU}(F) \geqs 1 - \varepsilon$.
\item (Soundness) For every assignment $F: \cV \to [R]$, $\val_{\cU}(F) \leqs \eta$.
\end{itemize}
\end{conjecture}

In other words, Khot's UGC states that it is NP-hard to distinguish between an UG instance which is almost satisfiable from one in which only small fraction of edges can be satisfied. While SSEH as stated in Conjecture~\ref{conj:sse} is not directly a statement about an UG instance, it has a strong connection with the UGC. Raghavendra and Steurer~\cite{RS10}, in the same work in which they proposed the conjecture, observes that SSEH is implied by a variant of UGC in which the soundness is strengthened so that the constraint graph is also required to be a small-set expander (i.e. every small set has near perfect edge expansion). In a subsequent work, Raghavendra, Steurer and Tulsiani~\cite{RST12} showed that the two conjectures are in fact equivalent. More formally, the following variant of UGC is equivalent to SSEH:

\begin{conjecture}[UGC with Small-Set Expansion (UGC with SSE)~\cite{RS10}] \label{conj:ugc-sse}
For every $\varepsilon, \eta > 0$, there exist $\delta = \delta(\varepsilon) > 0$ and $R = R(\varepsilon, \eta)$ such that, given an UG instance $(\cG = (\cV, \cE, \cW), [R], \{\pi_e\}_{e \in \cE})$ such that $\cG$ is regular, it is NP-hard to distinguish between the following two cases:
\begin{itemize}
\item (Completeness) There exists an assignment $F: \cV \to [R]$ such that $\val_{\cU}(F) \geqs 1 - \varepsilon$.
\item (Soundness) For every assignment $F: \cV \to [R]$, $\val_{\cU}(F) \leqs \eta$. Moreover, $\cG$ satisfies $\Phi(S) \geqs 1 - \varepsilon$ for every $S \subseteq \cV$ of size $\delta n$.
\end{itemize}
\end{conjecture}

While our result is based on SSEH (which is equivalent to UGC with SSE), Bhangale \etal~\cite{BGHKK16} relies on another strengthened version of the UGC, which requires the following additional properties:
\begin{itemize}
\item There is not only an assignment that satisfies almost all constraints, but also a partial assignment to almost the whole graph such that every constraint between two assigned vertices is satisfied.
\item The graph in the soundness case has to satisfy the following vertex expansion property: for every not too small subset of $\cV$, its neighborhood spans almost the whole graph.
\end{itemize} 

More formally, the conjecture can be stated as follows.

\begin{conjecture}[Strong UGC (SUGC)~\cite{BK09}] \label{conj:sugc}
For every $\varepsilon, \eta, \delta > 0$, there exists $R = R(\varepsilon, \eta, \delta)$ such that, given an UG instance $(\cG = (\cV, \cE, \cW), [R], \{\pi_e\}_{e \in \cE})$ such that $\cG$ is regular, it is NP-hard to distinguish between the following two cases:
\begin{itemize}
\item (Completeness) There exists a subset $S \subseteq \cV$ of size at least $(1 - \varepsilon)|\cV|$ and a partial assignment $F: S \to [R]$ such that every edge inside $S$ is satisfied.
\item (Soundness) For every assignment $F: \cV \to [R]$, $\val_{\cU}(F) \leqs \eta$. Moreover, $\cG$ satisfies $|\Gamma(S)| \geqs (1 - \delta)|\cV|$ for every $S \subseteq \cV$ of size $\delta n$ where $\Gamma(S)$ denote the set of all neighbors of $S$.
\end{itemize}
\end{conjecture}

The conjecture was first formulated by Bansal and Khot~\cite{BK09}. We note here that the name ``Strong UGC'' was not given by Bansal and Khot, but was coined by Bhangale \etal~\cite{BGHKK16}. In fact, the name ``Strong UGC'' was used earlier by Khot and Regev~\cite{KR08} to denote a different variant of UGC, in which the completeness is strengthened to be the same as in Conjecture~\ref{conj:sugc} but the soundness does not include the vertex expansion property. Interestingly, this variant of UGC is equivalent to the original version of the conjecture~\cite{KR08}. Moreover, as pointed out in~\cite{BK09}, it is not hard to see that the soundness property of SUGC can also be achieved by simply adding a complete graph with negligible weight to the constraint graph. In other words, both the completeness and soundness properties of SUGC can be achieved separately. However, it is not known whether SUGC is implied by UGC.

To the best of our knowledge, it is not known if one of Conjecture~\ref{conj:ugc-sse} and Conjecture~\ref{conj:sugc} implies the other. In particular, while the soundness cases of both conjectures require certain expansion properties of the graphs, Conjecture~\ref{conj:ugc-sse} deals with edge expansion whereas Conjecture~\ref{conj:sugc} deals with vertex expansion; even though these notations are closely related, they do not imply each other. Moreover, as pointed out earlier, the completeness property of SUGC is stronger than that of UGC with SSE; we are not aware of any reduction from SSE to UG that achieves this while maintaining the same soundness as in Conjecture~\ref{conj:ugc-sse}.

Finally, we note that both soundness and completeness properties of SUGC are crucial for  Bhangale \etal's reduction~\cite{BGHKK16}. Hence, it is unlikely that their technique applies to SSEH. Similarly, our reduction relies crucially on edge expansion properties of the graph and, thus, is unlikely to be applicable to SUGC.

\end{document}